\documentclass[11pt]{article}
\usepackage[margin=1in]{geometry}
\usepackage{amsmath,amssymb,amsthm}
\usepackage{graphicx}
\usepackage{booktabs}
\usepackage{algorithm}
\usepackage[noend]{algpseudocode}
\usepackage{subcaption}
\usepackage{hyperref}
\usepackage{xspace}
\usepackage{float}

\newtheorem{theorem}{Theorem}
\newtheorem{lemma}[theorem]{Lemma}

\newcommand{\Dom}{\mathcal{D}}
\newcommand{\pass}{\textup{\textsc{Passes}}\xspace}
\newcommand{\ent}[1]{\textsf{#1}}
\newcommand{\nd}{\texttt{native\_decide}\xspace}
\algrenewcommand\algorithmicrequire{\textbf{Input:}}
\algrenewcommand\algorithmicensure{\textbf{Output:}}

\title{Two-Colour Zig-Zag Numberlink on Rectangles:\\
A Forbidden-Pattern Characterization, a Constructive\\
Polynomial-Time Algorithm, and a Lean Verification}
\author{Z. Zzyzek\thanks{To the extent possible under law, the author has waived all copyright and
related or neighboring rights to this paper and to the accompanying code, proofs and data, under
CC0 1.0 Universal (\url{https://creativecommons.org/publicdomain/zero/1.0/}).}}
\date{\today}
\begin{document}
\maketitle

\begin{abstract}
In Zig-Zag Numberlink, pairs of endpoints in a grid must be joined by vertex-disjoint paths that
together visit every cell. The problem is NP-complete when the number of pairs is unbounded
(Adcock et al.). For a single pair it is the Hamiltonian path problem in rectangular grid graphs,
solved by Itai, Papadimitriou and Szwarcfiter. Two pairs on a rectangle form the paired 2-disjoint
path cover problem. We settle the case of two pairs on rectangles with the endpoints anywhere in the grid.
If both sides of the rectangle are at least $11$, an instance is solvable if and only if no entry of
a finite catalogue of forbidden patterns occurs. The catalogue consists of a parity condition, an
alternation condition on the border, a handful of local corner and edge patterns, $65$ listed
corner configurations, and a rule that propagates forced routes out of the corners before testing
alternation. If a side is at most $10$, a frontier dynamic programming algorithm decides the
instance in time linear in the long side. The characterization is constructive. Solutions are
built by deleting empty lines, splitting the rectangle by straight cuts, and splicing the pieces
back together, in time $O((R+C)\,RC)$. Both directions of the characterization and the correctness
of the dynamic programming algorithm are verified in Lean~4. The finite case analyses, which are
window certificates for the corner configurations and a check of all irreducible instances with
sides $11$--$22$, are carried out by the Lean compiler via \nd. We describe the algorithm, the
catalogue, the proof structure, the compromises made in the formalization, and implementations in
Python, JavaScript and C that produce identical output.
\end{abstract}

\section*{AI disclosure}
Nearly the entirety of the work for this project was done by Claude (Opus 5.5), including writing
this paper, writing the reference implementations and writing the Lean 4 proof. I, the listed
author, contributed little more than the problem statement and guidance. I have listed myself as
an author only to satisfy any requirement that a human author be attached to this work.

\section{Introduction}

Numberlink is a pencil puzzle. The player is given pairs of numbered cells in a grid and connects
each pair by a path along edge-adjacent cells, so that no two paths share a cell. The
\emph{zig-zag} variant, popular as the mobile game \emph{Flow Free}, additionally requires every
cell of the grid to be used. Adcock et al.~\cite{adcock15} showed that Zig-Zag Numberlink is
NP-complete when the number of pairs is unbounded.

\begin{figure}[t]
\centering
\includegraphics[width=0.32\textwidth]{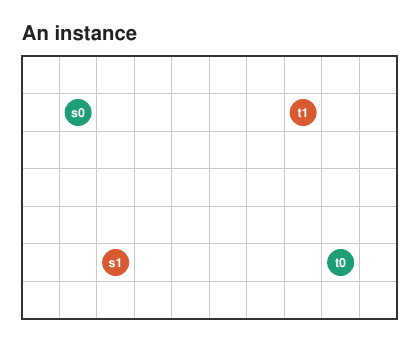}\hfill
\includegraphics[width=0.32\textwidth]{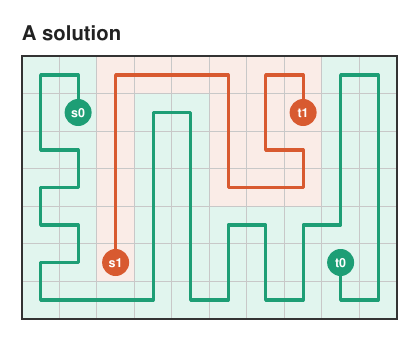}\hfill
\includegraphics[width=0.32\textwidth]{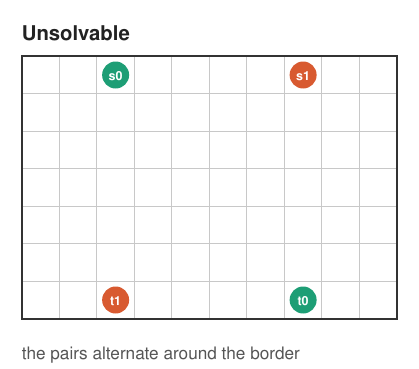}
\caption{Two-colour Zig-Zag Numberlink on a $7 \times 10$ rectangle. Left, an instance with
endpoints $s_0, t_0$ of colour~$0$ (green) and $s_1, t_1$ of colour~$1$ (orange). Middle, a
solution. The two vertex-disjoint paths $s_0 \to t_0$ and $s_1 \to t_1$ together visit every cell,
and the colour-$0$ path winds around the colour-$1$ path. Right, an unsolvable instance. The
endpoints alternate in colour around the border, so the two paths would have to cross (catalogue
entry \ent{T1}, Section~\ref{sec:catalogue}).}
\label{fig:splash}
\end{figure}

With a single pair, the problem asks for a Hamiltonian path between two given cells. Shade the
cells of the rectangle like a checkerboard. Itai, Papadimitriou and Szwarcfiter (IPS)~\cite{ips82}
showed that if the rectangle has an even number of cells, a Hamiltonian path between two cells
requires them to have different shades, and if it has an odd number of cells, both must have the
majority shade. A pair that meets this condition has a Hamiltonian path unless it belongs to one of
three small families of exceptions, which occur only when a side has length $1$, $2$ or~$3$.

We study the case where two pairs of endpoints lie on an $R \times C$ rectangle
(Figure~\ref{fig:splash}). The instance has endpoints $s_0, t_0$ of colour~$0$ and $s_1, t_1$ of
colour~$1$. A solution is a pair of vertex-disjoint paths $s_0 \to t_0$ and $s_1 \to t_1$ that
together cover every cell. In graph theory this is the \emph{paired 2-disjoint path cover}
problem. We allow the endpoints anywhere.

With two pairs, one path can block the other. In the right panel of Figure~\ref{fig:splash}, the
order of the endpoints around the border rules out a solution on a rectangle of any size. Near a
corner, a few endpoints can also force a path through the corner cells, and the forced route can
leave the remaining endpoints in alternating order (Section~\ref{sec:catalogue}). For rectangles
with both sides at least $11$, we show that four kinds of conditions account for every unsolvable
instance. These are the parity of the endpoints, their order around the border, a few patterns of
nearby endpoints, and a finite list of configurations near the corners.

\begin{theorem}[Main theorem]\label{thm:main}
Let $I$ be an instance on an $R \times C$ rectangle with $R, C \ge 11$. Then $I$ is solvable if and
only if \pass$(I)$, i.e.\ no entry of the catalogue of Section~\ref{sec:catalogue} fires. If
$\min(R, C) \le 10$, solvability is decided by a dynamic programming algorithm in time linear in
$\max(R, C)$. In both cases a solution, when one exists, is constructed in polynomial time.
\end{theorem}

Since the catalogue is a finite list of patterns, checking it takes time $O(R + C)$ when both sides
are at least $11$. When a side is at most $10$, the dynamic programming algorithm takes time
$O(\max(R, C))$, with a constant that depends only on the short side. Together they decide the
solvability of any instance in time $O(\max(R, C))$.

\paragraph{Contributions.}
\begin{itemize}
\item A finite catalogue of forbidden patterns that characterizes solvability for both sides
  $\ge 11$, with the endpoints anywhere in the grid (Section~\ref{sec:catalogue}).
\item A constructive algorithm. Line deletions are undone by a \emph{band extension} lemma that
  lifts solutions back, straight cuts (\emph{moves}) split the rectangle into pieces whose solutions
  glue together, and thin rectangles are solved by the plug DP (Sections~\ref{sec:algorithm}
  and~\ref{app:plugdp}).
\item A Lean~4 formalization of both directions and of the thin case (Section~\ref{sec:lean}),
  with the compromises made explicit.
\item Implementations in Python, JavaScript and C, tested against each other, the Lean definitions
  and exact solvers (Section~\ref{sec:impl}).
\end{itemize}

\section{Definitions}\label{sec:defs}

Cells are $(r, c)$ with $0 \le r < R$ and $0 \le c < C$. Two cells are adjacent if they differ by
one in exactly one coordinate. The instance is \emph{thin} if $\min(R, C) \le 10$, and the
\emph{domain} $\Dom$ is the set of instances with both sides $\ge 11$. Figure~\ref{fig:defs}
illustrates the definitions below.

\begin{figure}[t]
\centering
\begin{subfigure}[b]{0.17\textwidth}\centering
  \includegraphics[height=5.2cm]{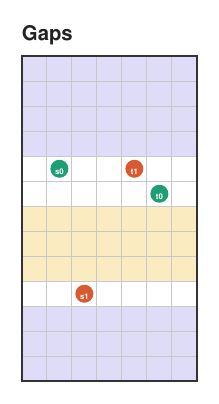}\caption{}\label{fig:defs:gaps}
\end{subfigure}\hfill
\begin{subfigure}[b]{0.32\textwidth}\centering
  \includegraphics[height=3.5cm]{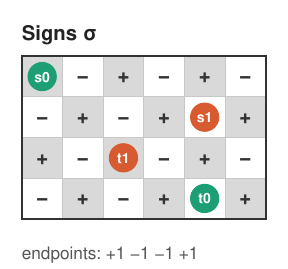}\caption{}\label{fig:defs:sigma}
\end{subfigure}\hfill
\begin{subfigure}[b]{0.22\textwidth}\centering
  \includegraphics[height=2.9cm]{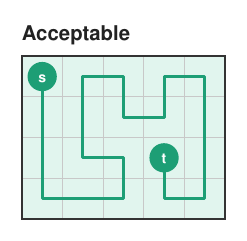}\caption{}\label{fig:defs:ipsyes}
\end{subfigure}\hfill
\begin{subfigure}[b]{0.22\textwidth}\centering
  \includegraphics[height=2.9cm]{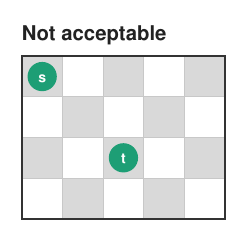}\caption{}\label{fig:defs:ipsno}
\end{subfigure}

\medskip
\begin{subfigure}[b]{0.4\textwidth}\centering
  \includegraphics[height=3.6cm]{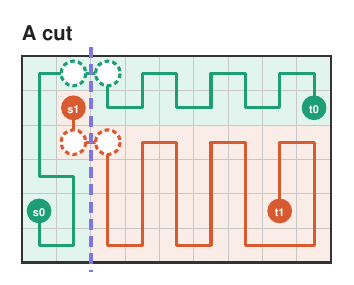}\caption{}\label{fig:defs:cut}
\end{subfigure}\hfill
\begin{subfigure}[b]{0.14\textwidth}\centering
  \includegraphics[height=3.6cm]{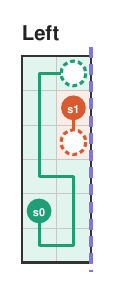}\caption{}\label{fig:defs:top}
\end{subfigure}\hfill
\begin{subfigure}[b]{0.36\textwidth}\centering
  \includegraphics[height=3.6cm]{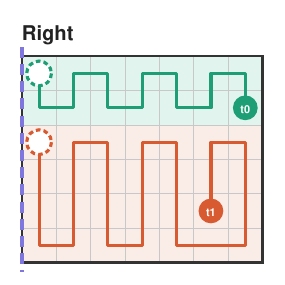}\caption{}\label{fig:defs:bottom}
\end{subfigure}
\caption{Definitions. (a) A $13 \times 7$ instance. Rows $0$--$3$ and $10$--$12$ are free and
form edge gaps (purple), and rows $6$--$8$ form an internal gap (yellow). (b) Black cells have sign
$\sigma = +1$ and white cells $-1$. The endpoint signs sum to $0$, which is $2 \cdot (RC \bmod 2)$
since $RC = 24$ is even, as parity requires. (c) A one-path piece. The cells $s$ and $t$ have
different shades on a $4 \times 5$ rectangle, so the pair is acceptable and a Hamiltonian $s$--$t$
path exists. (d) Both cells are black on a rectangle with an even number of cells, so the pair is
not acceptable. (e) A solution and a cut (dashed). Each colour crosses the cut once, and the two
crossing pairs are circled. (f), (g) The pieces left and right of the cut, with the cut drawn on
the side where they meet. Each crossing cell becomes a virtual endpoint (dashed) of its piece, and
the pieces' solutions glue back into (e).}
\label{fig:defs}
\end{figure}

\paragraph{Lines and gaps.}
A \emph{line} is a full row or column. It is \emph{free} if it contains no endpoint. A \emph{gap}
is a maximal run of consecutive free lines on one axis. An \emph{edge gap} touches the border, and
an \emph{internal gap} lies between two lines that contain endpoints (Figure~\ref{fig:defs:gaps}).
The reductions of Section~\ref{sec:algorithm} delete two lines from a gap.

\paragraph{Shades and signs.}
A cell is \emph{black} if $r + c$ is even and \emph{white} otherwise. Its sign $\sigma(x)$ is $+1$ if
$x$ is black and $-1$ if white (Figure~\ref{fig:defs:sigma}). We say shade rather than colour, to
keep \emph{colour} for the two pairs. Black and white cells alternate along any path, so a path
from $s$ to $t$ has sign sum $(\sigma(s) + \sigma(t))/2$. The whole grid has sign sum
$RC \bmod 2$. Hence a solution requires $\sum \sigma(\text{endpoint}) = 2 \cdot (RC \bmod 2)$, which
is the parity entry \ent{P} of the catalogue.

\paragraph{One-path pieces.}
For a rectangle and two cells $s \ne t$, IPS~\cite{ips82} show that a Hamiltonian $s$--$t$ path
exists if and only if two conditions hold. If $RC$ is even, $s$ and $t$ have different shades, and
if $RC$ is odd, both are black. In addition, $(s, t)$ is not in one of three explicit families that
occur when a side has length $1$, $2$ or $3$. We call such a pair \emph{acceptable}
(Figures~\ref{fig:defs:ipsyes} and~\ref{fig:defs:ipsno}). A Hamiltonian path between an acceptable
pair can be constructed in linear time.

\paragraph{Cuts and moves.}
A \emph{cut} is a straight line between two adjacent rows or columns. It splits the rectangle into
two \emph{pieces}. Where a solution path crosses the cut, the two cells on either side form a
\emph{crossing pair}, and each becomes a \emph{virtual endpoint} of its piece
(Figures~\ref{fig:defs:cut}--\ref{fig:defs:bottom}). The part of the path inside a piece then runs
between two endpoints, real or virtual. A piece is a \emph{one-path piece} if it holds the two
endpoints of a single path, and a \emph{two-path piece} otherwise. A \emph{move} is a cut together
with a choice of crossing pairs. It is \emph{valid} if every piece is OK, meaning that a one-path
piece is acceptable and a two-path piece is either thin and solvable or in $\Dom$ and passing the
catalogue. Solutions of the pieces then glue back along the crossing pairs.

\section{Related work}\label{sec:related}

\paragraph{Hamiltonian paths in grid graphs.}
Itai, Papadimitriou and Szwarcfiter~\cite{ips82} proved the Hamiltonian path problem NP-complete in
general grid graphs, which are finite induced subgraphs of the infinite grid and may have holes and
irregular boundaries. They also gave the characterization for rectangles that we use for one-path
pieces. This paper is about rectangles only. On a rectangle, long runs of free lines can be
shortened without changing solvability (Section~\ref{sec:algorithm}), and the remaining
obstructions come from parity, the order of the endpoints along the border, and endpoints near the
corners (Section~\ref{sec:catalogue}). General grid graphs have no such reduction, and for them the
problem is NP-complete already for one pair. For one pair on a rectangle the problem is exactly the
IPS case.

\paragraph{Numberlink.}
Adcock et al.~\cite{adcock15} proved that Zig-Zag Numberlink is NP-complete. For a single pair the
problem is the Hamiltonian path problem, solved on rectangles by IPS. Without the covering
requirement, the problem is a vertex-disjoint paths problem. That problem is NP-complete for
unbounded $k$ even in planar graphs~\cite{lynch75}, and polynomial for every fixed $k$ in general
graphs by the Graph Minors theory of Robertson and Seymour~\cite{rs95}. With the covering
requirement the paths must also fill the grid, so parity and corner effects become obstructions
that do not arise for disjoint paths alone.

\paragraph{Disjoint path covers.}
In graph theory, Zig-Zag Numberlink is the \emph{paired many-to-many $k$-disjoint path cover}
problem, which asks for $k$ vertex-disjoint paths that join given pairs of vertices and together
cover every vertex. It has been studied for many interconnection networks, such as hypercubes,
tori, and cylindrical and toroidal grids~\cite{parkihm15}. Park and Ihm~\cite{parkihm15} solved the
case $k = 2$ on bipartite cylindrical grids, which are rectangles with wraparound along the rows.
Above a small size, the only obstructions there are parity and three configurations of the
endpoints, and these coincide with our entries \ent{T1}, \ent{T2} and \ent{E1}. A cylinder has no
corners, and the rest of our catalogue concerns the corners of the rectangle.
Keshavarz-Kohjerdi~\cite{kk24} stated a characterization for the case where all four endpoints lie on
the border of the rectangle.

\paragraph{Frontier dynamic programming algorithms.}
The dynamic programming algorithm we use for thin rectangles is a transfer-matrix (frontier-based)
method over the connectivity states of a boundary. In competitive programming it is known as the
\emph{plug DP}. The standard reference is Chen's 2008 paper~\cite{chen08}, written in Chinese and
summarized in~\cite{oiwiki}. It introduced the bracket encoding of the frontier that we use, and it
handles the endpoints of a path as unpaired (``independent'') plugs. In statistical physics the same
technique enumerates Hamiltonian circuits, walks and chains on lattices of bounded
width~\cite{jacobsen07}. Our DP is this technique with each plug labelled by the colour of its path,
and we do not claim novelty for it. Our contribution regarding the DP is its formal verification in
Lean (Section~\ref{sec:lean}).

\section{Overview of the approach}\label{sec:overview}

Thin rectangles, those with $\min(R, C) \le 10$, are handled exactly by the plug DP
(Appendix~\ref{app:plugdp}). For rectangles with both sides at least $11$, every entry of the
catalogue other than parity looks only at the endpoints near the four corners, at the order of the endpoints along
the border, or at small groups of nearby endpoints. The proof has three parts.
\begin{enumerate}
\item \emph{Obstructions} (Theorem~B). Each entry of a catalogue of patterns is proved to prevent a
  solution. Most proofs are short and use parity, the planarity of the two paths, or forced degrees
  at corners. Some corner configurations are certified by exhaustive local computations.
\item \emph{Reductions} (Theorem~A, induction step). If an axis has a long run of free lines, two of
  them are deleted. The smaller instance passes the catalogue if the larger one does, because each
  entry sees the same corner windows and the same cyclic order of endpoints along the border. Any
  solution of the smaller instance extends to the larger one by \emph{band extension}.
\item \emph{Base} (Theorem~A, base case). A side of length $\ge 23$ always admits a reduction, so
  the irreducible instances have both sides between $11$ and $22$. For each passing irreducible
  instance there is a valid move, a straight cut whose pieces are smaller and OK. This finite
  statement is checked by computer.
\end{enumerate}
Theorems A and B together give the characterization, and the reductions and moves give the
construction.

\section{The catalogue of forbidden features}\label{sec:catalogue}

The catalogue is closed under the eight symmetries of the rectangle, under swapping the colours,
and under swapping $s$ with $t$ within a colour. We list entries up to these symmetries. Corner
entries are evaluated in the eight \emph{corner frames} (four corners, two orientations), with
frame coordinates $(i, j)$ counted from the corner. Below, $\gamma$ is either colour and
$\bar\gamma = 1 - \gamma$ is the other one. Entry names are set in sans-serif (\ent{R}, \ent{C}) to
distinguish them from the side lengths $R$ and $C$.

\begin{table}[ht]
\centering
\small
\begin{tabular}{@{}lp{0.72\textwidth}@{}}
\toprule
Entry & Fires when \\
\midrule
\ent{P} & $\sum \sigma(\text{endpoint}) \ne 2 \cdot (RC \bmod 2)$ (the parity of two covering paths) \\
\ent{T1} & all endpoints are on the border and their colours alternate $0, 1, 0, 1$ around it \\
\ent{T2} & the endpoints fill a $2\times2$ square with each colour on a diagonal \\
\ent{L1} & an endpoint's in-grid neighbours are all endpoints of the other colour \\
\ent{L6} & two corners are closed off, one by each colour, each being an empty corner cell whose two
    neighbours are endpoints of one colour \\
\ent{L2}--\ent{L5} & small corner patterns, for example \ent{L2} with $(0,1)$ of colour $\gamma$,
    $(1,0)$ of colour $\bar\gamma$, and $(0,0)$ empty \\
\ent{E1} & along an edge, $(0,\ell)$ and $(1,\ell{+}1)$ have colour $\gamma$ and $(1,\ell{+}2)$ and
    $(0,\ell{+}3)$ have colour $\bar\gamma$ \\
\ent{B} & a listed corner triple occurs ($21$ triples, or $5$ when both sides are odd), with the
    fourth endpoint on the border outside the $6\times6$ corner window \\
\ent{C}, \ent{D} & the four endpoints form one of $36$ listed corner configurations ($RC$ even) or
    $3$ ($RC$ odd) within an $8\times8$ corner window \\
\ent{R} & after forced corner routes (rules R1--R3), the effective endpoints alternate along the
    border of the remaining region (Figure~\ref{fig:R}) \\
\bottomrule
\end{tabular}
\caption{The catalogue. The lists for \ent{B}, \ent{C} and \ent{D} are in Appendix~\ref{app:lists}.}
\label{tab:catalogue}
\end{table}

\begin{figure}[ht]
\centering
\includegraphics[height=3.3cm]{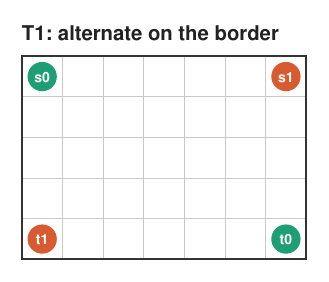}\hfill
\includegraphics[height=3.3cm]{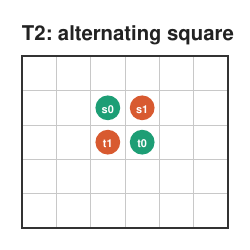}\hfill
\includegraphics[height=3.0cm]{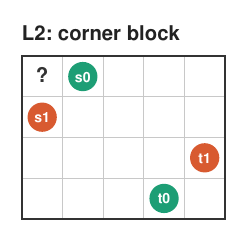}

\medskip
\includegraphics[height=3.0cm]{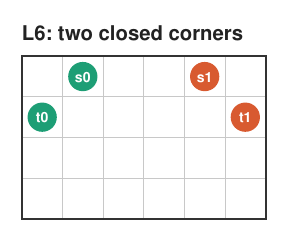}\qquad
\includegraphics[height=4.6cm]{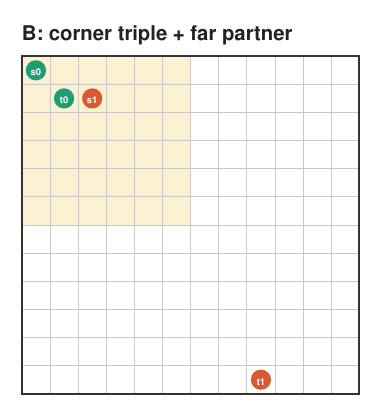}
\caption{Entries \ent{T1}, \ent{T2}, \ent{L2}, \ent{L6} and \ent{B}. In \ent{T2} the endpoints fill
a $2\times2$ square with each colour on a diagonal, so the two paths would have to cross inside the
square. In \ent{L2} the corner cell (marked ?) can only be reached from two endpoints of different
colours. In \ent{B} the $6\times6$ corner window is shaded.}
\label{fig:catalogue}
\end{figure}

\begin{figure}[t]
\centering
\begin{subfigure}[b]{0.15\textwidth}\centering
  \includegraphics[height=2.6cm]{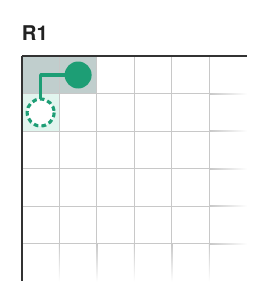}\caption{}
\end{subfigure}\hfill
\begin{subfigure}[b]{0.15\textwidth}\centering
  \includegraphics[height=2.6cm]{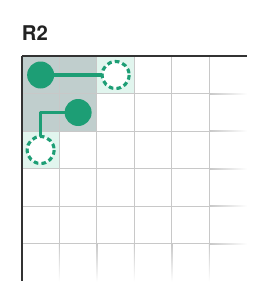}\caption{}
\end{subfigure}\hfill
\begin{subfigure}[b]{0.15\textwidth}\centering
  \includegraphics[height=2.6cm]{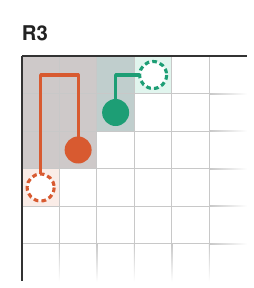}\caption{}
\end{subfigure}\hfill
\begin{subfigure}[b]{0.25\textwidth}\centering
  \includegraphics[width=\linewidth]{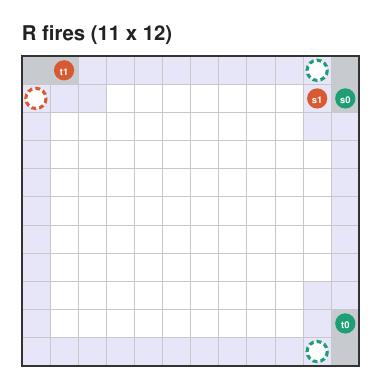}\caption{}
\end{subfigure}\hfill
\begin{subfigure}[b]{0.25\textwidth}\centering
  \includegraphics[width=\linewidth]{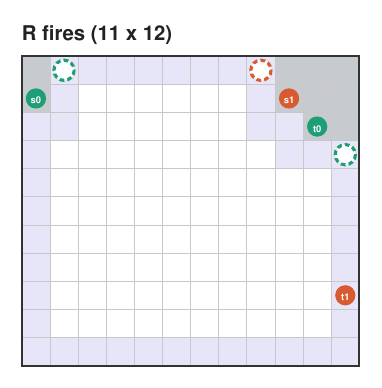}\caption{}
\end{subfigure}
\caption{The \ent{R} entry. (a)--(c) Rules R1--R3 at the top-left corner of the grid, which
continues to the right and below. Green is a colour $\gamma$ and orange the other colour
$\bar\gamma$. The forced path pieces are drawn, settled cells are grey, and dashed circles are the
effective endpoints. (d), (e) Two $11 \times 12$ instances on which \ent{R} fires and no other
entry does. Neither has a solution. After the forced routes, the four effective endpoints lie on
the border of the remaining region (shaded purple) and alternate in colour along it. This is the
\ent{T1} condition for that region, so the two paths would have to cross. In (d) rule R1 applies at
three corners, and in (e) rule R3 settles eight cells at the top-right corner and R1 settles two at
the top-left.}
\label{fig:R}
\end{figure}

\paragraph{The \ent{R} entry.}
Some corner configurations force a path through the corner cells (Figure~\ref{fig:R}). The rules
below are stated in a corner frame.
\begin{itemize}
\item[R1] An endpoint at $(0,1)$, with $(0,0)$ and $(1,0)$ not endpoints. The corner cell $(0,0)$
  has only two neighbours, so the path runs $(0,1) \to (0,0) \to (1,0)$. The endpoint effectively
  moves to $(1,0)$, and $(0,0)$ and $(0,1)$ are settled.
\item[R2] Endpoints of one colour at $(0,0)$ and $(1,1)$, with $(0,1)$, $(1,0)$, $(0,2)$ and
  $(2,0)$ not endpoints. The paths are forced around the $2\times2$ corner square, and the
  endpoints effectively move to $(0,2)$ and $(2,0)$.
\item[R3] Endpoints of different colours at $(1,2)$ and $(2,1)$, with $(0,0)$--$(0,3)$, $(1,0)$,
  $(1,1)$, $(2,0)$ and $(3,0)$ not endpoints. The two paths fill the corner, the endpoints
  effectively move to $(0,3)$ and $(3,0)$, and eight cells are settled.
\end{itemize}
The rules are applied once per corner frame in a fixed order (at each frame R3, then R2, then R1),
skipping a rewrite that would overlap cells already settled. \ent{R} fires if at least one rewrite
applied and the four effective endpoints are distinct cells on the boundary of the remaining region
whose colours alternate along it. This is \ent{T1} for the remaining region.

\paragraph{Where the entries come from.}
\ent{P} and \ent{T1} are the classical obstructions. Each path alternates black and white cells,
which gives parity, and two disjoint paths between alternating border points would cross, which
gives planarity. \ent{T2}, \ent{L1}--\ent{L6} and \ent{E1} express forced degrees near endpoints
and corners. \ent{B}, \ent{C} and \ent{D} are corner configurations that are unsolvable for reasons
that are local but not short. Each was found by exhaustive search on corner windows, and each is
certified by a window computation (Section~\ref{sec:proof}).

\section{The algorithm}\label{sec:algorithm}

\subsection{Reductions and band extension}

On an axis of length $n \ge 13$, the following reductions delete two free lines $d$ and $d+1$.
\begin{itemize}
\item[(S)] A \emph{strip} applies to an edge gap of at least $4$ lines and deletes the outermost two.
\item[(K)] A \emph{compression} applies to an internal gap of at least $10$ lines and deletes two
  from its middle.
\item[(I)] An \emph{interior compression} applies to an internal gap of at least $3$ lines
  containing $d$ and $d+1$ with $8 \le d \le n - 10$.
\end{itemize}

\begin{lemma}[Band extension]\label{lem:band}
If $I'$ is solvable and row $r$ of $I'$ is free, then inserting two empty rows directly below $r$
gives a solvable instance $I$. A solution of $I$ is obtained from one of $I'$ in time $O(RC)$.
\end{lemma}
\begin{proof}[Construction]
Stretch every path edge crossing below row $r$ through the two new rows. The new rows split into
blocks between consecutive crossings, and each crossing absorbs the block on one of its sides by a
U-shaped detour. The remaining block is absorbed by a \emph{square flip}, in which a path edge
$(r, a)$--$(r, a{+}1)$ inside it is replaced by the block's boundary cycle. Such an edge exists by
counting. The $m$ crossing cells of row $r$ supply at most $m$ horizontal edges, while each of the
$m+1$ runs of non-crossing cells needs one, because its cells have degree two and no vertical edge
downward.
\end{proof}

\begin{figure}[ht]
\centering
\includegraphics[height=3.2cm]{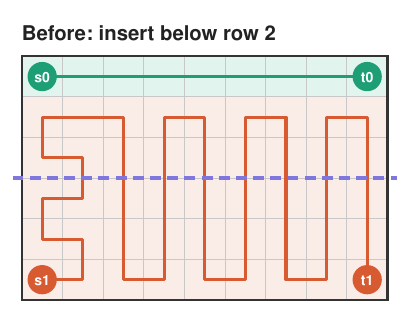}\qquad
\includegraphics[height=3.8cm]{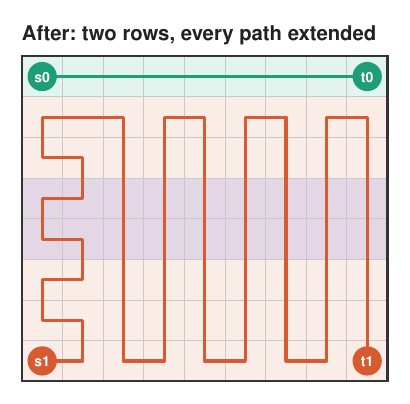}
\caption{Band extension. Two rows are inserted below row 2.}
\label{fig:band}
\end{figure}

\subsection{Moves}

\begin{figure}[ht]
\centering
\includegraphics[height=4.3cm]{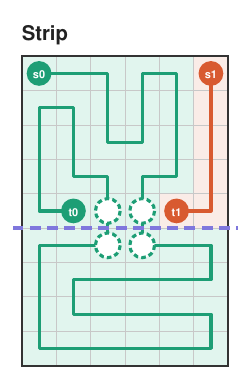}\hfill
\includegraphics[height=4.3cm]{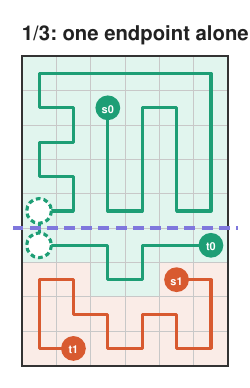}\hfill
\includegraphics[height=4.3cm]{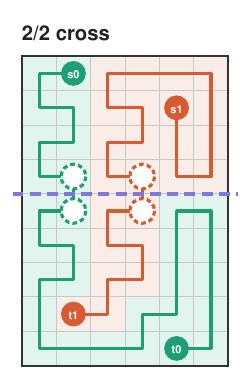}\hfill
\includegraphics[height=4.3cm]{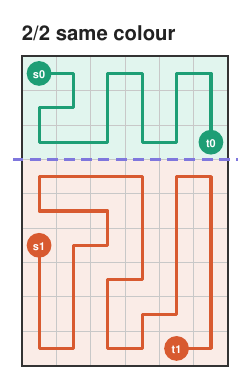}\hfill
\includegraphics[height=4.3cm]{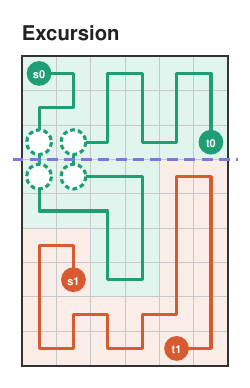}
\caption{The five move types. Each cuts the rectangle along the dashed line, and the dashed circles
mark crossing pairs, whose cells become virtual endpoints of the pieces. In a \emph{strip} all four
endpoints lie on one side, and the other side is an empty block of even area, spliced into a path
that runs along the cut. In a \emph{1/3} move one endpoint is alone on its side. Its path crosses
the cut once, so that side is a one-path piece from the endpoint to its crossing cell, and the
other side is a two-path piece. In a \emph{2/2 cross} each side holds one endpoint of each colour.
Each path crosses once, and both sides are two-path pieces. In a \emph{2/2 same} move each side
holds both endpoints of one colour. No path crosses, and both sides are one-path pieces. In an
\emph{excursion} each side holds one whole pair as in 2/2 same, but one path (here colour~$0$)
leaves its side and returns, crossing twice.}
\label{fig:moves}
\end{figure}

A move cuts the rectangle between rows $p-1$ and $p$, up to symmetry. The five types differ in how
the endpoints fall (Figure~\ref{fig:moves}). In a \emph{strip} all endpoints are on one side, and
the other side is empty, has even area and at least two lines. Its block is spliced into a path
running along the cut. In a \emph{1/3} move one endpoint is alone, and its piece is a one-path
piece to its crossing. A \emph{2/2 cross} has one endpoint of each colour on each side and two
crossings. In a \emph{2/2 same} move each side holds one whole pair, giving two one-path pieces. An
\emph{excursion} also has one whole pair on each side, and one colour crosses the cut twice. The
pieces' solutions glue along the crossing edges.

\subsection{Pseudocode}

Algorithm~\ref{alg:decide} decides an instance, Algorithm~\ref{alg:solve} constructs a solution, and
Algorithm~\ref{alg:band} is the band extension used to undo a reduction. The catalogue test
\textsc{Fires} is written out in Appendix~\ref{app:fires} (Algorithm~\ref{alg:fires}). Each entry's
test is a constant-size pattern check (Table~\ref{tab:catalogue} and Appendix~\ref{app:lists}),
except for \ent{T1}, \ent{E1} and \ent{R}, which scan the border. The test takes time $O(R + C)$.

\begin{algorithm}[tb]
\caption{Decision}\label{alg:decide}
\begin{algorithmic}[1]
\Require instance $I = (R, C, s_0, t_0, s_1, t_1)$ with four distinct endpoints
\Ensure whether $I$ is solvable
\If{$\min(R, C) \le 10$} \Return \Call{PlugDPFeasible}{$I$} \Comment{Appendix~\ref{app:plugdp}}
\Else\ \Return \textbf{not} \Call{Fires}{$I$} \Comment{Section~\ref{sec:catalogue}}
\EndIf
\end{algorithmic}
\end{algorithm}

\begin{algorithm}[tb]
\caption{Construction}\label{alg:solve}
\begin{algorithmic}[1]
\Function{Solve}{$I$}
  \If{$\min(R, C) \le 10$} \Return \Call{PlugDPSolve}{$I$} \EndIf
  \If{\Call{Fires}{$I$}} \Return \textsc{None} \EndIf
  \If{$(\mathit{axis}, d) \gets$ \Call{FindReduction}{$I$} succeeds}
    \State $I' \gets I$ with lines $d, d+1$ on $\mathit{axis}$ deleted
    \State $S' \gets$ \Call{Solve}{$I'$}
    \State \Return \Call{BandExtend}{$S'$, $\mathit{axis}$, $d$} \Comment{Lemma~\ref{lem:band}}
  \EndIf
  \If{a cut separates the two pairs and both one-path pieces are acceptable}
    \State \Return the two Hamiltonian paths (IPS construction)
  \EndIf
  \State $M \gets$ \Call{FindMove}{$I$} \Comment{exists by Theorem~A, and here $R, C \le 22$}
  \State solve each piece of $M$ (by IPS, the plug DP, or \Call{Solve}{})
  \State \Return the pieces' solutions glued along $M$'s crossings
\EndFunction
\Statex
\Function{FindReduction}{$I$}
  \For{$\mathit{axis} \in \{\text{rows}, \text{columns}\}$ with length $n \ge 13$}
    \If{lines $0..3$ are free} \Return $(\mathit{axis}, 0)$ \EndIf
    \If{lines $n{-}4..n{-}1$ are free} \Return $(\mathit{axis}, n-2)$ \EndIf
    \If{lines $\ell..\ell{+}9$ are free, with endpoints on both sides} \Return $(\mathit{axis}, \ell+4)$ \EndIf
    \If{$8 \le d \le n{-}10$, lines $d, d{+}1$ and $d{-}1$ or $d{+}2$ free, endpoints on both sides}
      \State \Return $(\mathit{axis}, d)$
    \EndIf
  \EndFor
  \State \Return \textsc{None}
\EndFunction
\Statex
\Function{FindMove}{$I$}
  \For{$m = 0, 1, \dots, 10$} \Comment{width of the widest thin piece, $0$ for none}
    \For{each move type, each of the 8 orientations, each endpoint relabelling}
      \For{each cut position $p$ and crossing choice whose widest thin piece has width $m$}
        \If{every piece is OK (narrower piece checked first)} \Return the move \EndIf
      \EndFor
    \EndFor
  \EndFor
\EndFunction
\end{algorithmic}
\end{algorithm}

\begin{algorithm}[tb]
\caption{Band extension (Lemma~\ref{lem:band}), rows}\label{alg:band}
\begin{algorithmic}[1]
\Require a solution $P$ of $I'$ and a free row $r$
\Ensure a solution of $I'$ with two empty rows $n_1 = r{+}1$, $n_2 = r{+}2$ inserted
\State shift every cell below row $r$ down by two
\State $X \gets$ columns $c$ where $P$ uses the edge $(r, c)$--$(r{+}3, c)$ \Comment{crossings}
\For{$c \in X$} replace the edge by $(r,c), (n_1,c), (n_2,c), (r{+}3,c)$ \EndFor
\State let $B_0, \dots, B_{|X|}$ be the column blocks of rows $n_1, n_2$ between crossings
\State $j \gets$ an index with $B_j = \emptyset$, or else one whose block contains a path edge
       $(r,a)$--$(r,a{+}1)$
\For{the $i$-th crossing $c$} absorb $B_{i-1}$ if $i \le j$, else $B_i$, by a U-detour replacing
     $(n_1,c)$--$(n_2,c)$ \EndFor
\If{$B_j \ne \emptyset$} replace $(r,a)$--$(r,a{+}1)$ by the boundary cycle of $B_j$ opened at
     $(n_1,a)$--$(n_1,a{+}1)$ \EndIf
\end{algorithmic}
\end{algorithm}

\subsection{Complexity}
Each reduction removes two lines, so at most $(R + C)/2$ reductions occur. Each costs a catalogue
check ($O(R + C)$) and a band extension ($O(RC)$). Moves occur only on rectangles with both sides
at most $22$, in constant time. The total is $O((R + C)\,RC)$ for $R, C \ge 11$, plus the plug DP,
which is linear in the long side for a fixed short side.

\subsection{A worked example}

\begin{figure}[ht]
\centering
\includegraphics[height=4.4cm]{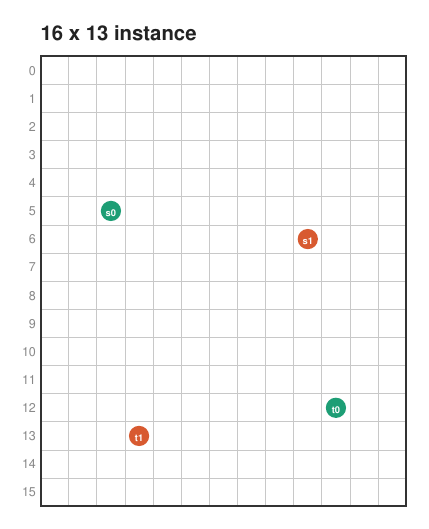}\hfill
\includegraphics[height=4.0cm]{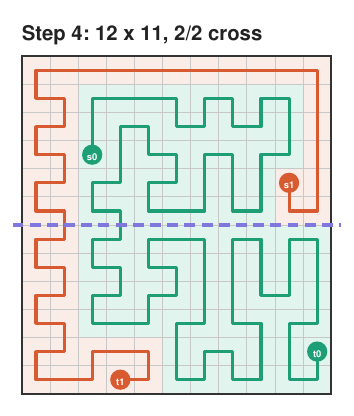}\hfill
\includegraphics[height=4.4cm]{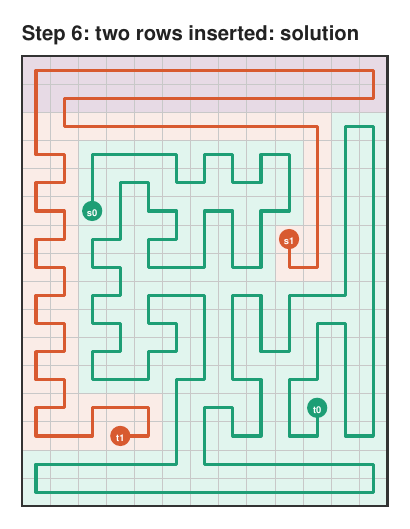}
\caption{A $16 \times 13$ instance, the $12\times11$ core with its 2/2-cross move, and the
solution.}
\label{fig:example}
\end{figure}

Consider the $16\times13$ instance of Figure~\ref{fig:example}. Rows $0$--$3$ are free, so the strip
reduction deletes rows $0$ and $1$. The $14\times13$ result has no reduction, and strip moves cut
off a $2\times13$ block below and a $12\times2$ block on the right. The $12\times11$ core is split by
a horizontal 2/2-cross move into two thin $6\times11$ pieces, which the plug DP solves. The blocks
are spliced back, and the two deleted rows are reinserted by band extension.

\section{Proof outline}\label{sec:proof}

\begin{theorem}[Theorem B, soundness of the catalogue]
If an entry fires on an instance in $\Dom$, the instance is unsolvable.
\end{theorem}
\ent{P}, \ent{T1}, \ent{T2}, \ent{L1}--\ent{L6}, \ent{E1} and \ent{R} have direct proofs, using
parity, the fact that two disjoint paths between alternating border points of a disc must cross,
and forced degrees. \ent{B}, \ent{C} and \ent{D} are certified by \emph{window certificates}. A
dynamic programming algorithm over the corner window enumerates all ways a solution can look inside
it, and an outside test shows that none of them can be completed outside the window, for any
rectangle in $\Dom$. The outside test uses parity, a non-crossing perfect matching of the boundary
crossings, and the structure of the paths.

\begin{theorem}[Theorem A, completeness of the catalogue]
Every instance in $\Dom$ that passes the catalogue is solvable.
\end{theorem}
The proof is by induction on the area. If a reduction applies, the smaller instance passes because
each entry's view is unchanged, so it is solvable by induction, and band extension lifts its
solution. Otherwise both sides are at most $22$, since a side $\ge 23$ always has a reduction. For
these finitely many instances a computer check shows that a valid move exists. Its thin pieces are
solvable, its other two-path pieces are smaller passing instances in $\Dom$ and solvable by
induction, and the pieces glue.

\section{Implementations}\label{sec:impl}

We provide libraries in Python (the reference), JavaScript and C, each with a command-line program.
All three implement Algorithms~\ref{alg:decide}--\ref{alg:band} and the plug DP with the same
iteration orders, and produce byte-identical output, including the paths. They were tested as
follows.
\begin{itemize}
\item The catalogue of each implementation agrees with the Lean definition on $60{,}000$ random
  instances with sides $2$--$24$ and endpoints clustered near corners.
\item The three implementations agree on $450$ random instances with sides $3$--$45$, and every
  reported solution is verified.
\item The C implementation solves $20{,}000$ random instances with sides $11$--$60$ in about
  $8$ seconds, with \textsc{Decide} and \textsc{Solve} agreeing on every instance.
\item On $300$ random instances with sides $11$--$12$, \textsc{Decide} agrees with an independent
  exact solver.
\end{itemize}
Each library also offers an optimized solver for thin instances. It cuts long thin rectangles into
pieces at guessed splits and falls back to the plug DP when a guess fails. It is exact, and on
random $10$-wide instances it is typically three orders of magnitude faster.

\section{Verification in Lean}\label{sec:lean}

The development is about $27{,}000$ lines of Lean~4~\cite{lean4} with Mathlib~\cite{mathlib20}.
The trusted statement is a short file defining instances, paths and solutions. The main results
are the following.
\begin{itemize}
\item \texttt{passes\_iff\_solvable} states that for $R, C \ge 11$ and a well-formed instance,
  $\pass(I) \leftrightarrow \mathrm{Solvable}(I)$.
\item \texttt{thinBF\_iff} states that the proof-oriented plug DP accepts if and only if the
  instance is solvable, for every rectangle. It uses only the standard axioms.
\item \texttt{solvableB\_iff} states that a decision procedure for every rectangle is correct.
\end{itemize}
Here \pass is a computable Boolean function. The theorem is about it as defined, so no informal
transcription of the catalogue needs to be trusted. The IPS characterization is formalized as well
and used for one-path pieces.

\paragraph{Compromises.}
\begin{enumerate}
\item \emph{Computation by \nd.} The window certificates ($10$ theorems) and the base case
  ($144$ theorems, one per box of sides $11$--$22$) are evaluated by the Lean compiler rather than
  the kernel. This adds $154$ axioms that trust the compiler and the generated native code. Kernel
  evaluation is not practical. On a $6\times6$ window the kernel exceeded $16$\,GB after
  $20$~minutes on a computation the compiled code does in $3$~seconds, and an $8\times8$ window
  reaches $2.3\times10^5$ states.
\item \emph{Symmetry handling.} The base case checks one representative per orbit of the
  $64$-element symmetry group (geometric symmetries, colour swap, reversing either path). Instead
  of proving that the catalogue is invariant under reflections, the computation checks that the
  representative or one of its reflections passes. Invariance under the other symmetries is
  proved. The search region is about $2.8\times10^{10}$ instances, and the check takes about
  $16$ CPU-hours ($4$ hours on $4$ threads).
\item \emph{Checked search instead of verified search.} The move search is a Boolean function
  proved sound, so a move it finds is valid. Thin pieces are solved by an unverified solver, and the
  solutions are checked by a verified checker.
\item \emph{Precompiled modules.} Without \texttt{precompileModules}, \nd interprets the
  project's functions and is about $400\times$ slower.
\item \emph{Running time.} The polynomial bound is not formalized, since Lean has no standard cost
  model short of Turing machines. The theorems certify what is computed, and the running-time
  analysis is on paper.
\item \emph{Thin case.} The verified plug DP uses a proof-friendly state encoding, with sorted lists
  of pieces coded by their ends, rather than the bracket encoding of the implementations.
\end{enumerate}

\section{Discussion}

The characterization gives a linear-time decision procedure for two pairs on rectangles with both
sides at least $11$, while the problem is NP-complete for an unbounded number of pairs. Open
questions include three or more pairs on rectangles, where corner interactions multiply, and two
pairs on other regions such as solid grid graphs. The corner configurations of the catalogue
(\ent{B}, \ent{C}, \ent{D}) were found by computer. A conceptual description of them would shorten
the proof and might allow kernel-checked certificates.

\bibliographystyle{plain}
\bibliography{references}

@article{ips82,
  author  = {Alon Itai and Christos H. Papadimitriou and Jayme Luiz Szwarcfiter},
  title   = {Hamilton Paths in Grid Graphs},
  journal = {SIAM Journal on Computing},
  volume  = {11},
  number  = {4},
  pages   = {676--686},
  year    = {1982},
  doi     = {10.1137/0211056}
}

@article{adcock15,
  author  = {Aaron B. Adcock and Erik D. Demaine and Martin L. Demaine and Michael P. O'Brien and
             Felix Reidl and Fernando S{\'a}nchez Villaamil and Blair D. Sullivan},
  title   = {Zig-Zag {Numberlink} is {NP}-Complete},
  journal = {Journal of Information Processing},
  volume  = {23},
  number  = {3},
  pages   = {239--245},
  year    = {2015},
  note    = {arXiv:1410.5845}
}

@article{rs95,
  author  = {Neil Robertson and Paul D. Seymour},
  title   = {Graph Minors. {XIII}. The Disjoint Paths Problem},
  journal = {Journal of Combinatorial Theory, Series B},
  volume  = {63},
  number  = {1},
  pages   = {65--110},
  year    = {1995},
  doi     = {10.1006/jctb.1995.1006}
}

@article{lynch75,
  author  = {James F. Lynch},
  title   = {The Equivalence of Theorem Proving and the Interconnection Problem},
  journal = {ACM SIGDA Newsletter},
  volume  = {5},
  number  = {3},
  pages   = {31--36},
  year    = {1975},
  doi     = {10.1145/1061425.1061430}
}

@inproceedings{lean4,
  author    = {Leonardo de Moura and Sebastian Ullrich},
  title     = {The {Lean} 4 Theorem Prover and Programming Language},
  booktitle = {Automated Deduction -- CADE 28},
  series    = {Lecture Notes in Computer Science},
  volume    = {12699},
  pages     = {625--635},
  year      = {2021},
  doi     = {10.1007/978-3-030-79876-5_37}
}

@inproceedings{mathlib20,
  author    = {{The mathlib Community}},
  title     = {The {Lean} Mathematical Library},
  booktitle = {Proceedings of the 9th ACM SIGPLAN International Conference on Certified Programs
               and Proofs (CPP)},
  pages     = {367--381},
  year      = {2020},
  doi     = {10.1145/3372885.3373824}
}

@misc{chen08,
  author       = {Danqi Chen},
  title        = {Dynamic Programming Problems Based on Connectivity State Compression},
  howpublished = {Paper of the Chinese National Olympiad in Informatics (NOI) national training team},
  note         = {In Chinese},
  year         = {2008}
}

@misc{oiwiki,
  author       = {{OI Wiki contributors}},
  title        = {Plug {DP}},
  howpublished = {\url{https://oi-wiki.org/dp/plug/}},
  note         = {In Chinese; accessed October 2026},
  year         = {2026}
}

@article{jacobsen07,
  author  = {Jesper Lykke Jacobsen},
  title   = {Exact Enumeration of {Hamiltonian} Circuits, Walks and Chains in Two and Three Dimensions},
  journal = {Journal of Physics A: Mathematical and Theoretical},
  volume  = {40},
  number  = {49},
  pages   = {14667--14678},
  year    = {2007},
  doi     = {10.1088/1751-8113/40/49/003}
}

@article{kk24,
  author  = {Fatemeh Keshavarz-Kohjerdi},
  title   = {Paired Many-to-Many 2-Disjoint Path Covers in Meshes},
  journal = {IEEE Transactions on Parallel and Distributed Systems},
  volume  = {35},
  number  = {10},
  pages   = {1854--1866},
  year    = {2024},
  doi     = {10.1109/TPDS.2024.3445283}
}

@article{parkihm15,
  author  = {Jung-Heum Park and Insung Ihm},
  title   = {Many-to-Many Two-Disjoint Path Covers in Cylindrical and Toroidal Grids},
  journal = {Discrete Applied Mathematics},
  volume  = {185},
  pages   = {168--191},
  year    = {2015}
}

\appendix

\section{The plug DP}\label{app:plugdp}

Transpose so that the grid has $W \le L$ rows. Cells are processed column by column, top to bottom.
The \emph{frontier} after some cells consists of $W + 1$ edges between processed and unprocessed
cells, one per row plus the edge below the current cell (Figure~\ref{fig:frontier}). A state
assigns each frontier edge a 3-bit code. The edge is unused, or it is an \emph{anchor} of colour $0$
or $1$, which is a piece ending at an endpoint, or it is an \emph{open} or \emph{close} bracket of
colour $0$ or $1$. A bracket pair is a piece with both ends on the frontier, and the pairs match like
parentheses because pieces cannot cross. This is the bracket encoding of Chen~\cite{chen08}, with
each plug labelled by its colour.

\begin{figure}[H]
\centering
\includegraphics[height=3.2cm]{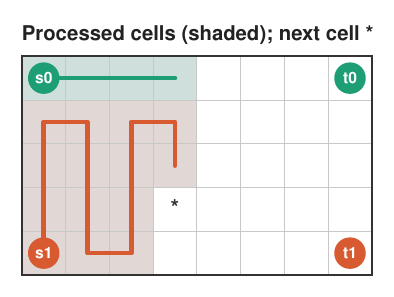}
\caption{Processed cells (shaded) and the next cell ($*$).}
\label{fig:frontier}
\end{figure}

\begin{algorithm}[H]
\caption{Plug DP (decision)}\label{alg:plugdp}
\begin{algorithmic}[1]
\Function{PlugDPFeasible}{$I$} \Comment{$W \le L$ after transposition}
  \State $S \gets \{\text{all-empty}\}$
  \For{column $c = 0, \dots, L-1$}
    \If{$c > 0$} shift every state by one slot \EndIf
    \For{row $r = 0, \dots, W-1$}
      \State $S \gets \bigcup_{s \in S} \Call{Step}{s, r, c}$
    \EndFor
  \EndFor
  \State \Return all-empty $\in S$
\EndFunction
\Statex
\Function{Step}{$s, r, c$} \Comment{slot $r$ is the plug from above, $r{+}1$ the plug from the left}
  \State $b \gets 1$ if $(r, c)$ is an endpoint of colour $\gamma$, else $2$
  \State $h \gets$ the number of incoming plugs
  \If{$h > b$ or $b - h$ exceeds the free directions (right, down)} \Return $\emptyset$ \EndIf
  \If{$h = 0$}
    \State at an endpoint, start a new anchor of colour $\gamma$, going right or down
    \State otherwise start a new bracket pair of either colour, going right and down
  \ElsIf{$h = 1$}
    \State continue the plug right or down, or end it if $(r, c)$ is an endpoint of its colour
    \State (ending an anchor completes its path, and ending a bracket makes its partner an anchor)
  \Else \Comment{join two plugs of the same colour}
    \State two anchors complete the path, and an anchor and a bracket make the partner an anchor
    \State a close then an open merge the pieces, and an open then a close form a cycle (rejected)
    \State two opens or two closes change the bracket type of the outer partner
  \EndIf
  \State \Return the resulting states, with outgoing plugs in slot $r$ (right) and $r{+}1$ (down)
\EndFunction
\end{algorithmic}
\end{algorithm}

\paragraph{Construction.} The forward pass stores one checkpoint, the state set, per column. A
backward pass re-runs each column from its checkpoint, recording predecessors and keeping only
states whose processed slots agree with the known end-of-column state. It then traces back the
chosen edges. The number of states is bounded by the number of valid frontiers (about $10^5$ for
$W = 10$), so the running time is $O(L \cdot W \cdot \#\text{states})$.

\paragraph{Correctness.} By induction over the cells, the reachable states are exactly the
signatures of partial solutions on the processed cells. A signature records the used frontier
edges, their pairing, and the colours of pieces ending at endpoints. Hence the empty final state is
reachable if and only if a solution exists. Lean verifies this for an equivalent DP with a different encoding.

\section{The catalogue test}\label{app:fires}

Algorithm~\ref{alg:fires} tests the entries of Table~\ref{tab:catalogue} in the order used by the
implementations and the Lean definition, and returns the first one that fires.

\begin{algorithm}[H]
\caption{The catalogue test}\label{alg:fires}
\begin{algorithmic}[1]
\Function{Fires}{$I$} \Comment{the first entry that fires, or \textsc{None}}
  \If{$\sum \sigma(\text{endpoint}) \ne 2 \cdot (RC \bmod 2)$} \Return \ent{P} \EndIf
  \If{all endpoints are on the border and their colours alternate around it} \Return \ent{T1} \EndIf
  \If{the endpoints fill a $2\times2$ square, each colour on a diagonal} \Return \ent{T2} \EndIf
  \If{an endpoint's neighbours are all endpoints of the other colour} \Return \ent{L1} \EndIf
  \If{two corners are closed off by endpoints, one corner by each colour} \Return \ent{L6} \EndIf
  \For{each of the $8$ corner frames $f$}
    \If{a pattern \ent{L2}--\ent{L5} occurs at the corner, or \ent{E1} along the frame's first edge}
      \State \Return that entry
    \EndIf
    \If{a listed \ent{B} triple occurs in $f$ and the fourth endpoint is on the border outside the
        $6\times6$ window} \Return \ent{B} \EndIf
    \If{$R, C \ge 10$ and the endpoints form a listed \ent{C} ($RC$ even) or \ent{D} ($RC$ odd)
        configuration in $f$} \Return it \EndIf
  \EndFor
  \State apply rules R1--R3 at each corner frame (R3, R2, R1, without overlaps), giving
         effective endpoints $E'$ and settled cells $S$
  \If{a rule applied and $E'$ are four distinct cells on the boundary of the grid minus $S$, with
      colours alternating along it} \Return \ent{R} \EndIf
  \State \Return \textsc{None}
\EndFunction
\end{algorithmic}
\end{algorithm}

\section{The listed corner configurations}\label{app:lists}

Cells are written $ij$ for frame coordinates $(i, j)$.

\noindent\ent{\textbf{B}}, for $RC$ even. Each triple $(p\,q\,r)$ has colour $\gamma$ at $p$ and $q$
and colour $\bar\gamma$ at $r$.\\
(00 11 12), (01 02 11), (01 02 12), (01 03 11), (01 04 11), (01 04 12), (01 05 11), (01 12 11),
(01 12 13), (01 12 22), (01 12 31), (01 13 03), (01 21 13), (01 21 22), (01 21 31), (02 21 10),
(03 11 10), (04 11 10), (04 21 10), (04 21 12), (05 11 10).

\noindent\ent{\textbf{B}}, for $R$ and $C$ odd.\\
(00 11 12), (00 11 22), (01 02 11), (01 04 11), (04 11 10).

\noindent\ent{\textbf{C}}, for $RC$ even. Each entry is $\{$one pair $\mid$ the other pair$\}$.\\
\{01 03 $\mid$ 11 13\}, \{01 06 $\mid$ 07 11\}, \{01 06 $\mid$ 11 70\}, \{01 06 $\mid$ 12 60\},
\{01 07 $\mid$ 11 60\}, \{01 13 $\mid$ 11 12\}, \{01 13 $\mid$ 12 22\}, \{01 13 $\mid$ 12 31\},
\{01 22 $\mid$ 11 12\}, \{01 22 $\mid$ 12 31\}, \{01 22 $\mid$ 13 30\}, \{01 70 $\mid$ 11 60\},
\{02 11 $\mid$ 21 30\}, \{02 20 $\mid$ 12 21\}, \{03 21 $\mid$ 31 40\}, \{06 21 $\mid$ 12 60\},
\{12 22 $\mid$ 13 21\}, \{12 31 $\mid$ 13 21\}, \{01 12 $\mid$ 04 22\}, \{01 12 $\mid$ 04 31\},
\{01 12 $\mid$ 13 20\}, \{01 12 $\mid$ 13 40\}, \{01 12 $\mid$ 20 22\}, \{01 12 $\mid$ 22 40\},
\{01 12 $\mid$ 31 40\}, \{01 13 $\mid$ 03 04\}, \{01 13 $\mid$ 03 20\}, \{01 13 $\mid$ 03 40\},
\{01 21 $\mid$ 02 22\}, \{01 21 $\mid$ 02 31\}, \{01 21 $\mid$ 04 13\}, \{01 21 $\mid$ 04 22\},
\{01 21 $\mid$ 04 31\}, \{01 21 $\mid$ 13 40\}, \{01 21 $\mid$ 22 40\}, \{01 21 $\mid$ 31 40\}.

\noindent\ent{\textbf{D}}, for $RC$ odd.\\
\{00 01 $\mid$ 11 20\}, \{00 14 $\mid$ 02 13\}, \{00 23 $\mid$ 02 13\}.

\end{document}